\documentclass[11pt,letterpapert]{article}
\usepackage{graphicx}
\usepackage{amssymb}
\usepackage{epstopdf}
\usepackage{fancyhdr}
\usepackage{amsmath}
\usepackage{amsthm}
\usepackage{hyperref}
\usepackage{makeidx}
\usepackage{mathdesign}
\usepackage{color}
\usepackage{geometry}
\usepackage{soul}
\usepackage{nameref}
\definecolor{purple}{rgb}{.9,0,.9}

\makeatletter
\let\orgdescriptionlabel\descriptionlabel
\renewcommand*{\descriptionlabel}[1]{%
  \let\orglabel\label
  \let\label\@gobble
  \phantomsection
  \edef\@currentlabel{#1}%
  \let\label\orglabel
  \orgdescriptionlabel{#1}%
}
\makeatother

\newcommand{\Nat}{\mathbb{N}}

\newcommand{\Real}{\mathbb{R}}

\newcommand{\rfa}{\qquad {\rm for \ all}\ }
\newcommand{\rfae}{\qquad {\rm for \ almost \ every}\ }

\newcommand{\cE}{{\cal E}}
\newcommand{\cG}{{\cal G}}\newcommand{\cH}{{\cal H}}

\newcommand{\cV}{{\cal V}}\newcommand{\cW}{{\cal W}}
\newcommand{\cY}{{Y}}

\newcommand{\ba}{{\bf a}}\newcommand{\bb}{{\bf b}}
\newcommand{\be}{{\bf e}}

\newcommand{\bn}{{\bf n}}

\newcommand{\bt}{{\bf t}}\newcommand{\bu}{{\bf u}}
\newcommand{\bv}{{\bf v}}\newcommand{\bw}{{\bf w}}
\newcommand{\bA}{{\bf A}}
\newcommand{\bB}{{\bf B}}\newcommand{\bD}{{\bf D}}
\newcommand{\bE}{{\bf E}}\newcommand{\bG}{{\bf G}}
\newcommand{\bH}{{\bf H}}
\newcommand{\bM}{{\bf M}}

\newcommand{\bQ}{{\bf Q}}\newcommand{\bR}{{\bf R}}

\newcommand{\Lin}{\mathop{\rm Sym}}

\newcommand{\Sym}{\mathop{\rm Sym}}

\newcommand{\bbC}{\mathbb{C}}

\newcommand{\bpsi}{\boldsymbol{\psi}}
\newcommand{\ve}{\varepsilon}

\newtheorem{theorem}{Theorem}[section]
\newtheorem{lemma}[theorem]{Lemma}

\newtheorem{definition}[theorem]{Definition}
\newtheorem{proposition}[theorem]{Proposition}

\def\Xint#1{\mathchoice
{\XXint\displaystyle\textstyle{#1}}%
{\XXint\textstyle\scriptstyle{#1}}%
{\XXint\scriptstyle\scriptscriptstyle{#1}}%
{\XXint\scriptscriptstyle\scriptscriptstyle{#1}}%
\!\int}
\def\XXint#1#2#3{{\setbox0=\hbox{$#1{#2#3}{\int}$ }
\vcenter{\hbox{$#2#3$ }}\kern-.6\wd0}}

\def\dashint{\Xint-}

\newcommand{\beqn}{\begin{equation}}
\newcommand{\eeqn}{\end{equation}}

\title{Periodic homogenization and material symmetry\\ in linear elasticity}
\author{Mariya Ptashnyk and Brian Seguin}

\begin{document}
\date{}

\maketitle

\begin{center}
\emph{Dedicated to Walter Noll on the occasion of his 90$^{\,th}$ birthday.}
\end{center}

\begin{abstract}
\noindent Here homogenization theory is used to establish a connection between the symmetries of a periodic elastic structure associated with the microscopic properties of an elastic material and the material symmetries of the effective, macroscopic elasticity tensor.  Previous results of this type exist but here more general symmetries on the microscale are considered.  Using an explicit example, we show that it is possible for a material to be fully anisotropic on the microscale and yet the symmetry group on the macroscale can contain elements other than plus or minus the identity.  Another example demonstrates that not all material symmetries of the macroscopic elastic tensor are generated by symmetries of the periodic elastic structure.\\


\end{abstract}

%\tableofcontents

\section{Introduction}

Elastic composites span a wide range of materials, from geometerials and biological tissues to synthetic materials.  See, for example, Cherkaev and Kohn \cite{CK} and Jones \cite{Jones}.  Knowing the properties of such materials can lead to a better understanding of their behavior in their physical environment.  The properties of a composite are dictated by the properties of its constituents.  However, determining the properties of a composite from those of its constituents is a nontrivial matter.  Homogenization theory can be used to rigorously derive the macroscopic properties of a composite from the microscopic properties of its building blocks.  Here we are interested in the case where the composite consists of a periodic arrangement of linearly elastic materials.  In this instance, homogenization theory provides techniques for taking the limit as the size of the periodic microstructure goes to zero and obtain an macroscopic elasticity tensor associated with the composite.  Determining the macroscopic or effective elasticity tensor involves solving systems of elliptic partial differential equations on a representative unit cell domain associated with the periodic microstructure of the composite material.  See, for example, Cioranescu and Donato \cite{CD}, Jikov, Kozlov, and Oleinik \cite{JKO}, and Oleinik, Shamaev, and Yosifian \cite{OSY} for results on the homogenization of the equations of linear elasticity.

Besides having a microscopic structure, some important engineering and biological materials have symmetry properties.
%Along with microscopic structures, some symmetry properties are important characteristics of  engineering and biological materials. 
It is possible for the response of a material (not necessarily a composite) to be unaffected by a change in reference configuration.  This leads to the notion of material symmetry.  Roughly speaking, the material symmetry group consists of all transformations of the reference configuration that leave the response of the material unchanged.  A connection between material symmetry and microstructure is well-known as different material symmetry groups are connected with different crystalline structures.  For a list of which material symmetry group is appropriate to assume for a given crystalline structure see, for example, Coleman and Noll \cite{CN} or Gurtin \cite{Gurtin}.  Thus, it is reasonable to conjecture, in the context of homogenization theory, that the properties of a material on the microscopic scale yield information about the material symmetry group on the macroscale.  In this paper we show that if the elasticity tensor describing the properties of the microstructure satisfies an invariance condition involving an affine transformation that preserves volume, see equation \eqref{microsym} in Section~\ref{sect2}, then the gradient of this transformation is a material symmetry of the macroscopic elasticity tensor.  Previous results of this type were established by Jikov, Kozlov, and Oleinik \cite{JKO} and Alexanderian, Rathinam, and Rostamian \cite{ARR}, but these authors only considered transformations that are rotations or reflections about a fixed point.  Examples of transformations not considered before include a translation together with a rotation or a unimodular transformation that is not a rotation.  A detailed discussion of how previous results compare with what is established in this paper is given after Theorem~\ref{thmmain} in Section~\ref{sectresult}.

One of the advantages of knowing the material symmetry group of the macroscopic elasticity tensor is related to numerics.  The components of the macroscopic tensor relative to a basis are obtained by solving unit cell problems, which are systems of elliptic partial differential equations.  Taking into account the major and minor symmetries of the elasticity tensor, the 21 elasticities are determined by solving six unit cell problems.  Hence, information that reduces the number of unknown components can decrease the number of unit cell problems needed to be solved and, thus, save considerable computational time.  Knowledge of the material symmetry group yields this kind of information.  For example, if the macroscopic elasticity tensor possesses cubic symmetry,  the nine unknown elasticities can be computed by solving only two unit cell problems.

The outline of the paper is as follows.  In Section~\ref{sect2}, the concept of a periodic elastic structure and an appropriate notion of symmetry are introduced.  Section~\ref{sectresult} contains the main result connecting symmetries of the periodic elastic structure with the material symmetry group of the macroscopic elasticity tensor.  In Section~\ref{sectEx}, several examples of periodic elastic structures are given and the corresponding material symmetry groups are mentioned.  One example showns that it is possible for the constituents of a composite to be completely anisotropic and yet the macroscopic elasticity tensor can have a symmetry group containing elements other than plus or minus the identity.  Another example demonstrates that there can be elements of the material symmetry group of the macroscopic elasticity tensor that are not associated with any symmetry of the periodic elastic structure.  Finally, Section~\ref{sectCon} contains some concluding remarks.  For completeness, we include a derivation of the macroscopic elasticity tensor in the Appendix.

\section{Periodic elastic structures}\label{sect2}

In this section we introduce the concept of a periodic (linearly) elastic structure and discuss the appropriate notion of symmetry associated with this structure.  Such symmetries are generalizations of the transformations used when discussing material symmetry.

For the purposes of this paper it is useful to distinguish between points and vectors and, hence, work in a Euclidean point space $\cE$ with a vector space $\cV$ equipped with an inner-product.  Lowercase letters such as $y$ or $z$ will denote elements of $\cE$ and are points, while boldface lowercase letters such as $\ba$ and $\bb$ denote vectors.  Boldface uppercase letters, such as $\bH$, will denote linear mappings from $\cV$ to itself.  The inner-product on $\cV$ will be used to identify $\cV$ with its dual.  It follows that given a linear mapping $\bH$ from $\cV$ to itself, we can consider its transpose $\bH^\top$ also to be a linear mapping from $\cV$ to itself.% and, hence, the summation $\bL+\bL^\top$ makes sense.

Consider a parallelepiped $\cY$ in $\cE$.  The cases when $\cE$ is two- or three-dimensional are of primary interest, however here it is assumed that $\cE$ is $n$-dimensional, with $n\in\Nat$.  The set $\cY$ can be used to form a periodic tessellation of $\cE$.  Let $\bb_i$, $i=1,\dots,n$, be the linearly independent vectors in $\cV$ that are the length of and parallel to the edges of $\cY$; see Figure~\ref{figtesseqmicro}(a).  For all $z\in\cE$, there are unique integers $k_i$, $i=1,\dots,n$, such that $z-\sum_{i=1}^nk_i \bb_i$ is an element of $Y$.  Thus, we can define %a point $\{z\}_Y$ in $Y$ can be obtained by
\beqn\label{unity}
\{z\}_Y:=z-\sum_{i=1}^nk_i \bb_i\in Y.
\eeqn

Here we are interested in linearly elastic materials that have a periodic structure.  To this end, consider a position dependent elasticity tensor $\bbC$ defined on $\cY$ and extended $Y$-periodically to all of $\cE$ so that
\beqn\label{perext}
\bbC(z)=\bbC(\{z\}_\cY)\rfa z\in\cE.
\eeqn
The pair $(\cY,\bbC)$ is referred to as a periodic (linearly) elastic structure.  In such a construction, $\cY$ is called the unit cell.

Consider another parallelepiped $\hat\cY$ in $\cE$ and elasticity tensor $\hat\bbC$ defined on $\hat\cY$ so that $(\hat\cY,\hat\bbC)$ is another periodic elastic structure.  We say that the periodic elastic structures $(\cY,\bbC)$ and $(\hat\cY,\hat\bbC)$ are equivalent if the periodic extension of $\bbC$ relative to $\cY$ is equal to the periodic extension of $\hat\bbC$ relative to $\hat\cY$.  With the aid of \eqref{perext}, this can be expressed as
$$
\bbC(\{z\}_{Y})=\hat\bbC(\{z\}_{\hat Y})\rfa z\in \cE.
$$
It is possible for the two periodic elastic structures to be equivalent even when $Y\not=\hat Y$.  To see this, consider a periodic elastic structure $(\cY,\bbC)$ and let $\ba$ be a vector.  Extend $\bbC$ as in \eqref{perext} and define $\bbC_\ba$ to be its restriction to the set $\ba+\cY$, which is the translation of the set $Y$ by $\ba$.  The structure $(\ba+\cY,\bbC_\ba)$ is equivalent to $(\cY,\bbC)$.

Here we consider symmetries of a periodic elastic structure generated by affine transformations that preserve volume.  Such a transformation $h$ of $\cE$ has the form
\beqn\label{isometry}
h(z)=z_\circ+\ba+\bH(z-z_\circ)\rfa z\in\cE,
\eeqn
where $z_\circ$ is a point in $\cE$, $\ba$ is a vector, and $\bH$ is a unimodular linear mapping, so that $|\text{det}\,\bH|=1$.  To interpret $h$, first consider the case where $\bH$ is a rotation.  In that case, in \eqref{isometry} $\ba$ acts as a translation and $z_\circ$ is the point about which the rotation $\bH$ takes place.  When $\bH$ is not a rotation, a similar interpretation applies but instead of the space being rotated about $z_\circ$, it is transformed about $z_\circ$ by $\bH$.  Unless otherwise stated, when speaking of a symmetry $h$, it is assumed to have the form \eqref{isometry}.

Some explanation to why transformations of the form \eqref{isometry} is warranted.  When discussing material symmetry of a material at a point $z_\circ$, transformations of the form
\beqn\label{csym}
g(z)=z_\circ+\bG(z-z_\circ)
\eeqn
are considered.  See, for example, equation (50.6) on page 285 of Gurtin, Fried, and Anand \cite{GFA}.  It is assumed that $\bG$ is unimodular so that $g$ preserves volume.  This ensures that the mass density of the material at $z_\circ$ is preserved.  See page 11 of Epstein and El$\dot{\text{z}}$anowski \cite{EE}.  The transformations considered in \eqref{isometry} also include a translation by a vector $\ba$ since we are interested in symmetries of periodic structures, which, by definition, are invariant under certain translations.  One can think of the transformations of the form \eqref{isometry} as a merger of the transformations considered for material symmetry and those for periodic structures.

%Unless otherwise stated, when speaking of a symmetry $h$, it is assumed to have the form \eqref{isometry}.  The motivation for considering mappings of this form is that they are used when discussing material symmetry. 

Let $\Sym$ denote the set of all symmetric linear mappings from $\cV$ to itself.

\begin{figure}
\includegraphics[width=3in]{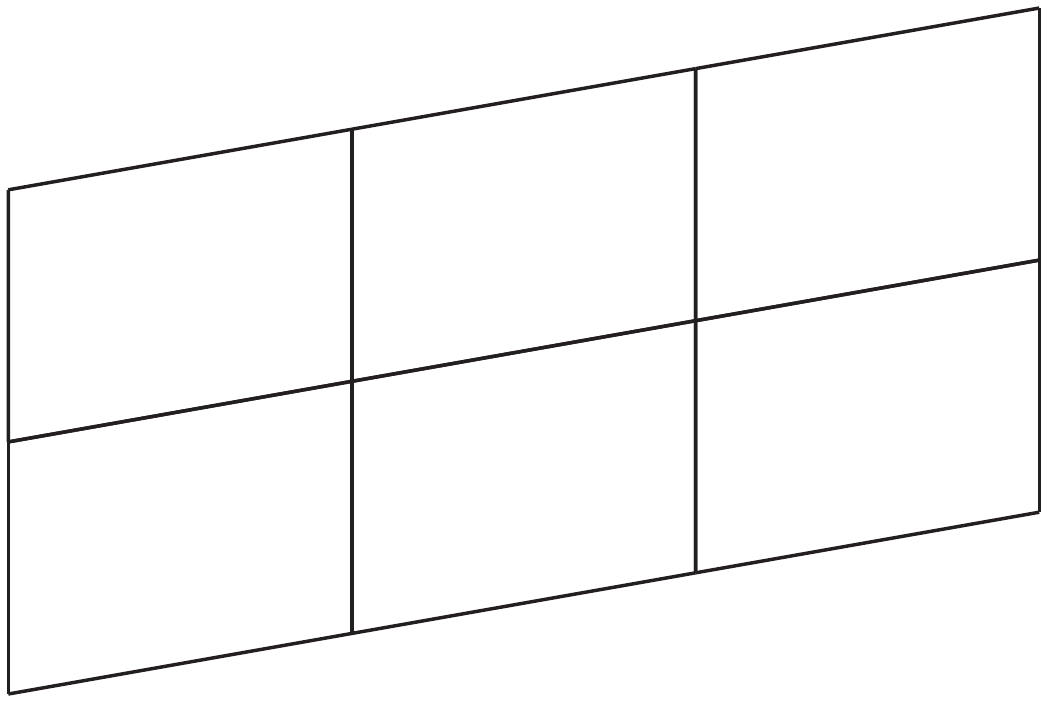}
\includegraphics[width=2.8in]{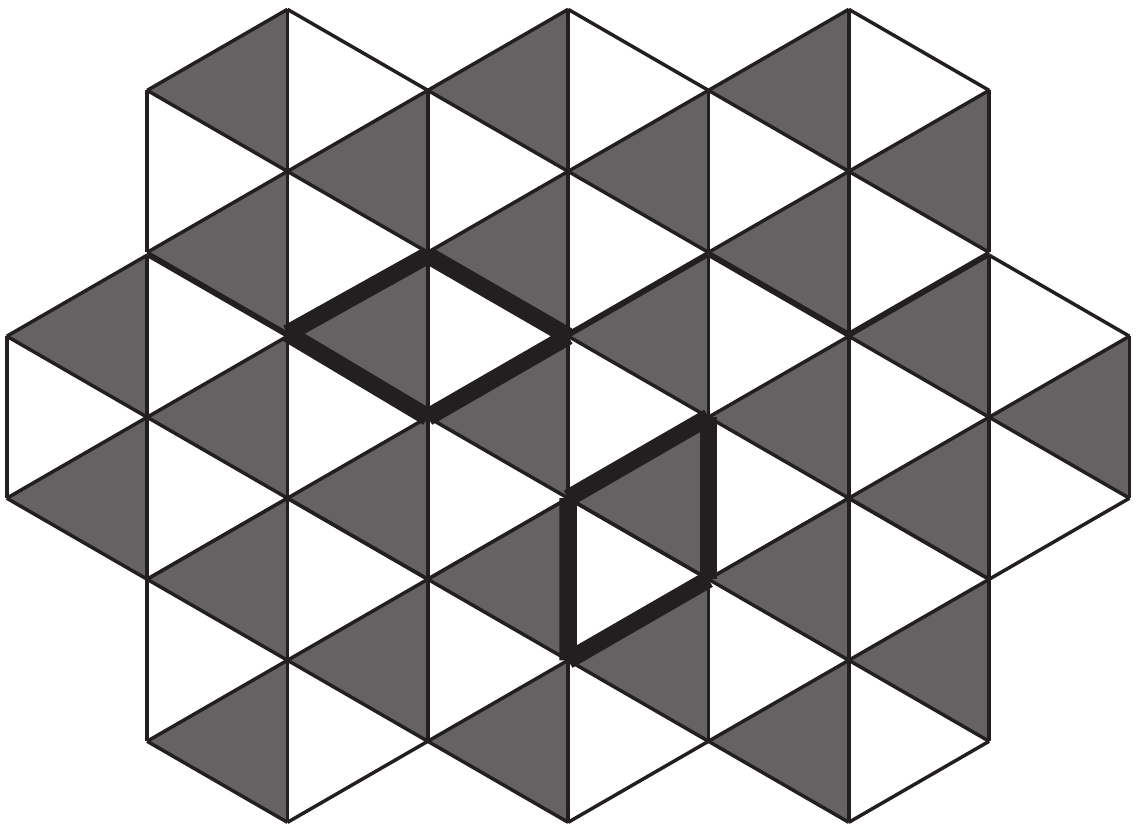}
\thicklines
\put(-315,45){$\cY$}
\put(-351,21.3){\rotatebox[origin=c]{10}{$\vector(1,0){72.5}$}}
\put(-310,12){$\bb_1$}
\put(-359.2,36){\rotatebox[origin=c]{90}{$\vector(1,0){52}$}}
\put(-365,35){$\bb_2$}
\put(-93,50){$\cY$}
\put(-131,83){$\hat\cY$}
\put(-65,100){$h$}
\put(-115,92){\rotatebox[origin=c]{-165}{$\vector(1,0){0}$}}
\qbezier(-73,63)(-55,125)(-112,102)
\put(-56,36){\rotatebox[origin=c]{30}{$\vector(1,0){28}$}}
\put(-61,39){\rotatebox[origin=c]{90}{$\vector(1,0){28}$}}
\put(-40,27){$\bb_1$}
\put(-64,38){$\bb_2$}
\put(-156,129){\rotatebox[origin=c]{150}{$\vector(1,0){28}$}}
\put(-150,115){\rotatebox[origin=c]{210}{$\vector(1,0){28}$}}
\put(-140,140){$\hat\bb_1$}
\put(-140,113){$\hat\bb_2$}
\put(-310,-15){(a)}
\put(-108,-15){(b)}
\caption{(a) A two-dimensional example of a unit cell $\cY$ and the periodic tessellation associated with it.  The vectors $\bb_1$ and $\bb_2$ are the length of and parallel to the sides of $\cY$ and generate the tessellation in that translating $\cY$ by all integer linear combinations of $\bb_1$ and $\bb_2$ covers the plane. (b) An example of a periodic elastic structure made up of two isotropic elastic materials that can be generated by two different unit cells $\cY$ and $\hat\cY$ that are related through $h$, which is a symmetry of the periodic elastic structure.}
\label{figtesseqmicro}
\end{figure}
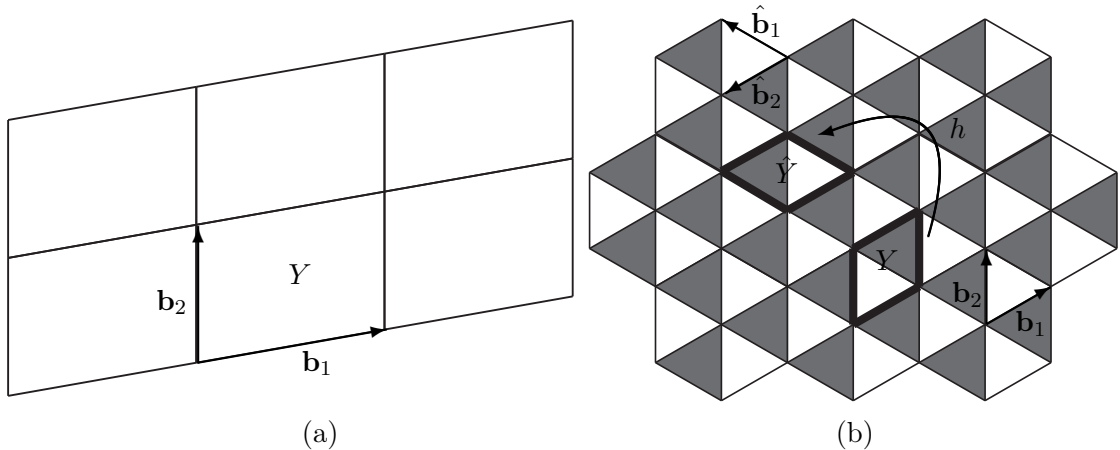

\begin{definition}\label{defmicrosym}
A transformation $h$ of the form \eqref{isometry} is a symmetry of the periodic elastic structure $(\cY,\bbC)$ if\,
\beqn\label{microsym}
\bbC(z)\bE= \bH[\bbC(h^{-1}(z))(\bH^\top\bE\bH)]\bH^\top\qquad \text{for all}\ z\in \cE,\ \bE\in\text{\rm Sym}.
\eeqn
\end{definition}

\noindent To understand the motivation behind Definition \ref{defmicrosym}, consider in \eqref{microsym} instead of $h$ the transformation $g$ defined in \eqref{csym} instead of $h$ in \eqref{microsym}.  Then \eqref{microsym} at $z=z_\circ$ would read
$$\bbC(z_\circ)\bE= \bG[\bbC(z_\circ)(\bG^\top\bE\bG)]\bG^\top\quad \text{for all}\ \bE\in\text{\rm Sym},$$
which would imply that the material point $z_\circ$ has $\bG$ as a material symmetry.  See, for example, equation (52.21) on page 300 of Gurtin, Fried, and Anand \cite{GFA}.  Thus, Definition \ref{defmicrosym} should be viewed as a generalization of the definition of material symmetry to periodic elastic structures.

It is easily seen that the collection of symmetries of a periodic elastic structure form a group under function composition.  For any unimodular linear mapping $\bH$, define
\beqn\label{notation}
\bbC_\bH(z)\bE:=\bH[\bbC(z)(\bH^{\small \top}\bE\bH)]\bH^\top\rfa z\in\cE,\ \bE\in\text{Sym}.
\eeqn
If $\bbC$ is periodic, then so is $\bbC_\bH$. Using the notation in \eqref{notation}, $h$ is a symmetry of the periodic elastic structure if and only if
\beqn\label{microsymN}
\bbC(z)= \bbC_\bH(h^{-1}(z))\qquad \text{for all}\ z\in \cE.
\eeqn
Notice that translations by integer linear combination of $\bb_i$, where $i=1,\dots,n$, are symmetries of the periodic elastic structure.

\begin{proposition}\label{symper}
If $h$ is a symmetry of the periodic elastic structure $(\cY,\bbC)$, then $\bbC$ is $\hat\bb_i$-periodic where $\hat\bb_i:=\bH\bb_i$.
\end{proposition}

\begin{proof}
Consider $\hat z\in\cE$ and set $z:=h^{-1}(\hat z)$.  From \eqref{microsymN} and the fact that $\bbC$ is $\bb_i$-periodic, we have
\begin{align*}
\bbC(\hat z+\hat\bb_i)=\bbC(h( z+\bb_i))=\bbC_{\bH}(z+\bb_i)=\bbC_{\bH}(z)=\bbC(\hat z),
\end{align*}
and so $\bbC$ is $\hat \bb_i$-periodic.
\end{proof}

The next result says that periodic elastic structures related by symmetries are equivalent.

\begin{proposition}\label{equivmicro}
Let $(\cY,\bbC)$ be a periodic elastic structure with symmetry $h$, and set $\hat \cY:=h(\cY)$. If $\hat\bbC$ is defined by
\beqn\label{barbbC}
\hat\bbC(\hat y):=\bbC_{\bH}(h^{-1}(\hat y))\rfa \hat y\in \hat \cY, 
\eeqn
then the periodic elastic structures $(\cY,\bbC)$ and $(\hat\cY,\hat\bbC)$ are equivalent.
\end{proposition}

\begin{proof}
It must be shown that the periodic extensions of $\bbC$ and $\hat\bbC$ relative to $\cY$ and $\hat \cY$, respectively, are equal.  Let $\hat z\in\cE$ be given and set $\hat\bb_i:=\bH\bb_i$, for $i=1,\dots,n$.  Define $\{\hat z\}_{\hat Y}$ analogous to \eqref{unity}, so that there are integers $\hat k_i$, $i=1,\dots,n$, such that
\beqn\label{unitybar}
\{\hat z\}_{\hat \cY} = \hat z - \sum_{i=1}^n\hat k_i \hat\bb_i\in\hat\cY.
\eeqn
By \eqref{barbbC} and the periodicity of $\bbC$ and $\hat\bbC$, see \eqref{perext}, and the definition of $\bbC_\bH$ in \eqref{notation}, we have
\beqn\label{cor1}
\hat\bbC(\hat z) = \hat\bbC(\{\hat z\}_{\hat\cY}) = \bbC_{\bH}(h^{-1}(\{\hat z\}_{\hat\cY}))= \bbC_{\bH}(\{h^{-1}(\{\hat z\}_{\hat\cY})\}_\cY).
\eeqn
Applying $h^{-1}$ to \eqref{unitybar} and using the definition of $\hat\bb_i$ yields
$$
h^{-1}(\{\hat z\}_{\hat \cY}) = h^{-1}(\hat z)-\sum_{i=1}^n\hat k_i\bb_i,
$$
from whence it follows that
$$
\{h^{-1}(\{\hat z\}_{\hat \cY})\}_\cY = \{h^{-1}(\hat z)\}_\cY.
$$
Using this last equality in \eqref{cor1}, the periodicity of $\bbC_{\bH}$, and equation \eqref{microsymN} results in
$$
\hat\bbC(\hat z) = \bbC_{\bH}(\{h^{-1}(\hat z)\}_\cY) = \bbC_{\bH}(h^{-1}(\hat z))=\bbC(\hat z).
$$
Since the above equation holds for any $\hat z\in \cE$, it follows that $\bbC=\hat\bbC$ and, hence, the elastic structures are equivalent.
\end{proof}

\noindent A consequence of Proposition~\ref{equivmicro} is that the elasticity tensors $\bbC$ and $\hat\bbC$ are equal in $\cE$, and so $(Y,\bbC)$ and $(\hat Y,\bbC)$ are equivalent elastic structures.  This fact will be used in the next section.

To illustrate the results of Propositions \ref{symper} and \ref{equivmicro}, consider the tessellation of the plane depicted in Figure~\ref{figtesseqmicro}(b).  Define $\bbC$ so that it equals one isotropic tensor on the white triangles and another isotropic tensor on the grey triangles.  This elastic structure can be generated by the unit cells $\cY$ and $\hat\cY$ shown in the figure.  The transformation $h$ that takes $\cY$ to $\hat\cY$ is a symmetry of the periodic elastic structure.  Moreover, the vectors $\hat\bb_1$ and $\hat\bb_2$ associated with the $\hat Y$-periodicity of the structure are related to the vectors $\bb_1$ and $\bb_2$ associated with the $Y$-periodicity of the structure through the gradient of $h$---that is, $\bb_i=\bH\hat\bb_i$, $i=1,2$.

%%%%%%%%%%%%%%%

\section{The macroscopic elasticity tensor}\label{sectresult}

Here we recall the macroscopic elasticity tensor of a composite material obtained by homogenizing the equations of linear elasticity associated with a composite with periodic microstructure.  After this, we prove our main result connecting the symmetries of the microstructure with the material symmetries of the macroscopic elasticity tensor.

Homogenization of the equations of linear elasticity is classical.  See, for example,  \cite{CD,JKO,OSY}.  For the sake of completeness, a sketch of the derivation of the macroscopic equations of linear elasticity in a domain with periodic microstructure is given in the Appendix.

Here we need to recall the formula for the macroscopic elasticity tensor.  The starting point is to consider a linearly elastic material with a periodic microstructure.  We consider the periodic elastic microstructure given by
\beqn\label{bbCscale}
\bbC^\ve(z)=\bbC\Big(\frac{z-q}{\ve}+q\Big)\rfa z\in\cE,
\eeqn
where $q$ is an arbitrary, but fixed, point, $\ve$ is a small parameter associated with the length scale of the microstructure, and $(Y,\bbC)$ is a periodic elastic structure.  If $\cY^\ve$ is the result of scaling $\cY$ by $\ve$ about the point $q$, then $(\cY^\ve,\bbC^\ve)$ is a periodic elastic structure, which can be viewed as a microstructure since $\ve$ is small; see Figure~\ref{figscale}.  

Homogenization theory provides techniques to derive an effective elasticity tensor $\bbC^0$ that describes a material whose behavior approximates that of a material with periodic microstructure $(\cY^\ve,\bbC^\ve)$ for small $\ve$.  The effective elasticity tensor $\bbC^0$ is given by
\beqn\label{hombbC}
\bbC^0\bE = \dashint_\cY \bbC(y)\big[\bE+\nabla_y\bw^\bE(y)\big]\, dy\rfa \bE\in \text{Sym},
\eeqn
where the symbol $\dashint$ denotes the average integral and\footnote{See the Appendix for the definitions of the function spaces $H^1_\text{per}(Y,\cV)$ and $\cW_\text{per}(Y,\cV)$.} $\bw^\bE\in \cW_\text{per}(Y,\cV)$ is the unique solution of
\begin{equation}\label{unitcell}
\displaystyle\int_\cY \bbC(y)\big[\bE+\nabla_y\bw^\bE(y)\big]\cdot\nabla_y\bv(y)\, dy = 0 \rfa \bv\in H^1_\text{per}(\cY,\cV).
\end{equation}
See \eqref{unitcellA} and \eqref{hombbCA} in the Appendix for more details.  Since $\bbC$ is an elasticity tensor, it possesses major and minor symmetries so that, in particular, for any linear mapping $\bA$, we have $\bbC\bA=\bbC\big[\frac{1}{2}(\bA+\bA^\top)\big]$.  Thus, only the symmetric part of $\nabla_{y}\bw^\bE$ has an influence in \eqref{hombbC} and \eqref{unitcell}.  %Given an orthonormal basis $\be_1,\dots,\be_n$ of $\cV$, \eqref{hombbC} can be written in components:
%\beqn\label{hombbCcomp}
%\bbC^0_{ijkl} = \dashint_\cY(\bbC_{ijkl}(y)+\sum_{p,q=1}^n\bbC_{ijpq}(y)\partial_{y_q}\bw^{\be_k\otimes\be_l}_{p}(y))\, dy.
%\eeqn

%
%Thus, \eqref{hom2} can be reformulated as follows: for any $\bE\in \Lin$, there is a unique $\bw^\bE\in \cW_\text{per}(Y,\cV)$ such that
%\begin{equation}\label{unitcell}
%\displaystyle\int_\cY \bbC(y)\big[\bE+\be_y(\bw^\bE)(y)\big]\cdot\be_y(\bv)(y)\, dy = 0 \rfa \bv\in H^1_\text{per}(\cY,\cV),
%\end{equation}
%which is called the unit cell problem.  The solutions of \eqref{unitcell} can be used to define
%\beqn\label{hombbC}
%\bbC^0\bE := \dashint_\cY \bbC(y)\big[\bE+\be_y(\bw^\bE)(y)\big]\, dy\rfa \bE\in \text{Sym},
%\eeqn
%which is called the macroscopic elasticity tensor.  Given an orthonormal basis $\be_1,\dots,\be_n$ of $\cV$, \eqref{hombbC} can be written in components:
%\beqn\label{hombbCcomp}
%\bbC^0_{ijkl} = \dashint_\cY(\bbC_{ijkl}(y)+\sum_{p,q=1}^n\bbC_{ijpq}(y)\partial_{y_q}\bw^{\be_k\otimes\be_l}_{p}(y))\, dy.
%\eeqn

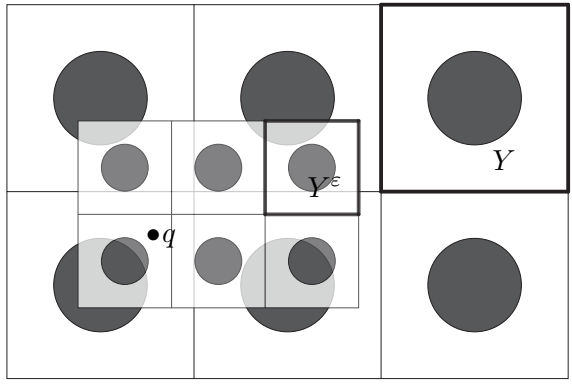
\begin{figure}
\hspace{1.5in}\includegraphics[height=2in]{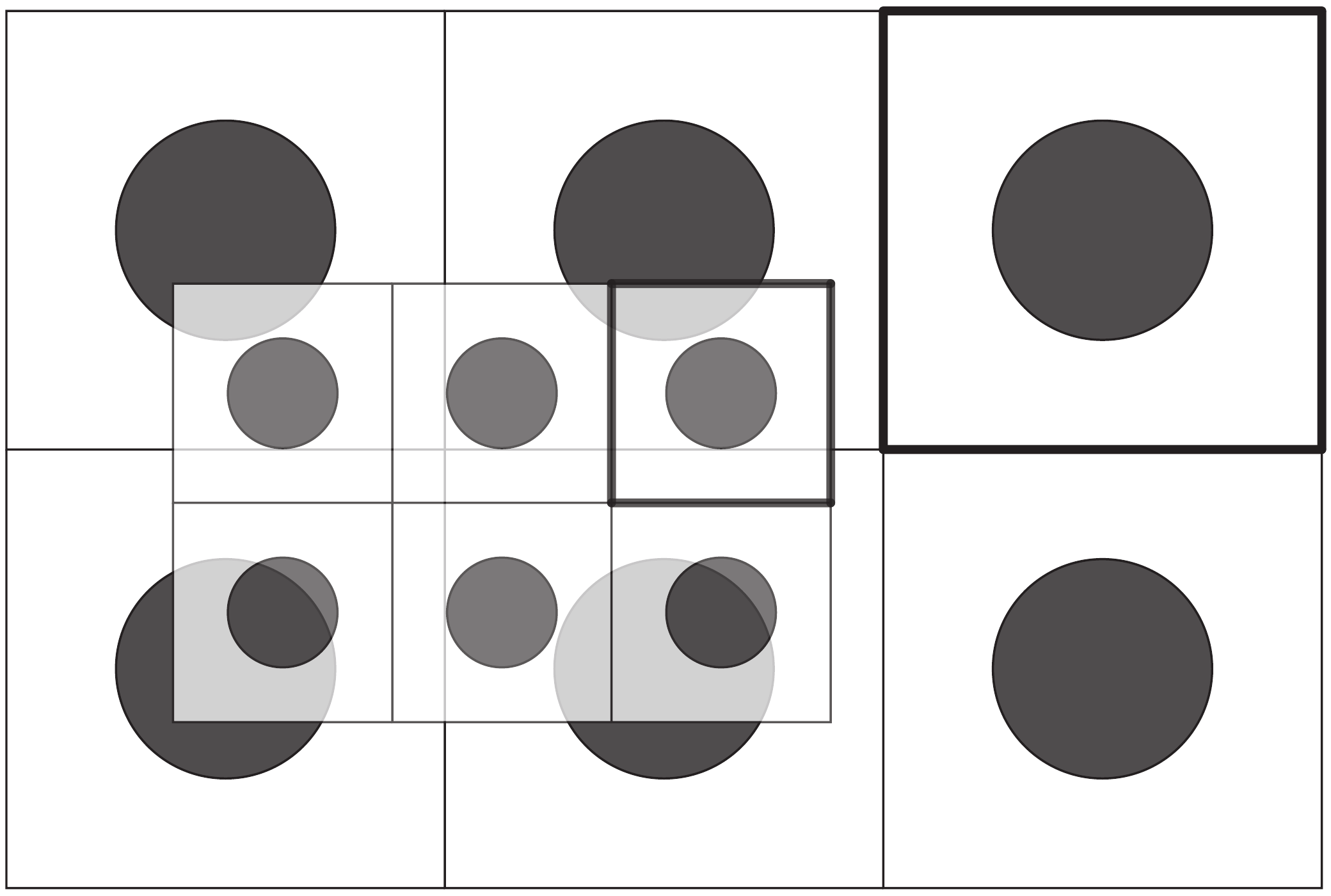}
\put(-30,80){$Y$}
\put(-100,70){$Y^\ve$}
\put(-161,53){$\bullet$}
\put(-155,53){$q$}
\caption{A depiction of a periodic elastic structure together with a scaling of it.  The unit cell $Y$ generates the initial structure, and the grey and white regions represent different elastic materials.  The smaller, transparent structure generated by the unit cell $Y^\ve$ is the result of scaling the original structure about the point $q$ by $\ve=1/2$.}
\label{figscale}
\end{figure}

The macroscopic elasticity tensor $\bbC^0$ can be given in terms of an equivalent periodic elastic structure with a different unit cell.  In particular, if $h$ is a symmetry of $(\cY,\bbC)$ and $\hat \cY:=h(\cY)$, then the remark made after Proposition~\ref{equivmicro} says that $(\hat Y,\bbC)$ is an equivalent elastic structure.  The homogenization procedure can be carried out using $\hat Y$ instead of $Y$ to find that $\bbC^0$ is also given by
\beqn\label{hombbCbar}
\bbC^0\bE = \dashint_{\hat \cY} \bbC(\hat y)\big[\bE+\nabla_{\hat y}\hat\bw^\bE(\hat y)\big]d\hat y\rfa \bE\in \Lin,
\eeqn
where $\hat\bw^\bE\in \cW_\text{per}(\hat Y,\cV)$ is the unique solution of
\begin{equation}\label{unitcellbar}
\displaystyle\int_{\hat \cY} \bbC(\hat y)\big[\bE+\nabla_{\hat y}\hat \bw^\bE(\hat y)\big]\cdot\nabla_{\hat y}\hat\bv(\hat y)\, d\hat y = 0 \rfa \hat\bv\in H^1_\text{per}(\hat \cY,\cV).
\end{equation}

The following result relates the solutions of the unit cell problems \eqref{unitcell} and \eqref{unitcellbar}.
\begin{lemma}\label{lemisow}
The solutions of the unit cell problems \eqref{unitcell} and \eqref{unitcellbar} satisfy
\beqn\label{bwrel}
\nabla_{y}(\bw^{\bH^\top\bE\bH})(h^{-1}(\hat y)) = \bH^\top\nabla_{\hat y}\hat\bw^\bE(\hat y)\bH\qquad \text{for almost every}\ \hat y\in\hat \cY=h(\cY).
\eeqn

\end{lemma}

\begin{proof}
Given $\hat \bv\in \cW_\text{per}(\hat Y,\cV)$, define
\beqn\label{barv}
\bv(y):=\bH^{\top}\hat\bv(h(y))\rfae y\in \cY.
\eeqn
Notice that since $\hat\bv\in\cW_\text{per}(\hat Y,\cV)$, then $\bv\in \cW_\text{per}(Y,\cV)$.  Hence, considering \eqref{unitcell} with $\bE$ replaced by $\bH^\top\bE\bH$ results in
$$
\int_\cY \bbC(y)\big[\bH^\top\bE\bH+\nabla_y(\bw^{\bH^\top\bE\bH})(y)\big]\cdot\nabla_y\bv(y)\, dy=0.
$$
Using the change of variables $\hat y=h(y)$ and the definition of $\bv$ in \eqref{barv} yields
\beqn\label{bla}
\int_{\hat \cY} \bbC(h^{-1}(\hat y))\big[\bH^\top\bE\bH+\nabla_y(\bw^{\bH^\top\bE\bH})(h^{-1}(\hat y))\big]\cdot\bH^{\top}\nabla_{\hat y}\hat \bv(\hat y)\bH\, d\hat y = 0.
\eeqn
Since $h$ is a symmetry of the elastic structure we have that $\bbC$ satisfies \eqref{microsym}, and so \eqref{bla} becomes
\beqn\label{bla2}
\int_{\hat \cY} \bbC(\hat y)\big[\bE+\bH^{-\top}\nabla_y(\bw^{\bH^\top\bE\bH})(h^{-1}(\hat y))\bH^{-1}\big]\cdot\nabla_{\hat y}\hat \bv(\hat y)\, d\hat y = 0.
\eeqn
Notice that $\bH^{-\top}\nabla_y(\bw^{\bH^\top\bE\bH})(h^{-1}(\hat y))\bH^{-1}=\nabla_{\hat y}(\bH^{-\top}\bw^{\bH^\top\bE\bH}\circ h^{-1})(\hat y)$ for almost every $\hat y\in Y$, and so \eqref{bla2} can be written as
\beqn\label{bla3}
\int_{\hat \cY} \bbC(\hat y)\big[\bE+\nabla_{\hat y}(\bH^{-\top}\bw^{\bH^\top\bE\bH}\circ h^{-1})(\hat y)\big]\cdot\nabla_{\hat y}\hat \bv(\hat y)\, d\hat y = 0.
\eeqn
Since \eqref{bla3} holds for all $\hat\bv\in \cW_\text{per}(\hat Y,\cV)$ and the solution of \eqref{unitcellbar} is unique, this proves the lemma.
\end{proof}

%%%%%%%%%%

Using Lemma~\ref{lemisow} we can now prove the main result of the paper.

\begin{theorem}\label{thmmain}
If $h$ is a symmetry of the periodic elastic structure, then
\beqn\label{mainresult}
\bbC^0\bE = \bH[\bbC^0(\bH^\top\bE\bH)]\bH^\top\rfa \bE\in\text{\rm Sym}.
\eeqn
\end{theorem}

\begin{proof}
Let $\bE\in \Sym$ be given and set $\hat\cY:=h(\cY)$.  By \eqref{hombbC}, the change of variables $\hat y=h(y)$, Lemma \ref{lemisow}, the definition of a symmetry of a periodic elastic structure \eqref{microsym}, and \eqref{hombbCbar}, we have
\begin{align*}
\bbC^0(\bH^{\top}\bE\bH) &= \dashint_\cY \bbC(y)\big[\bH^\top\bE\bH+\nabla_y(\bw^{\bH^\top\bE\bH})(y)\big]dy\\
&= \dashint_{\hat\cY} \bbC(h^{-1}(\hat y))\big[\bH^\top\bE\bH+\nabla_y(\bw^{\bH^\top\bE\bH})(h^{-1}(\hat y))\big]d\hat y\\
&= \dashint_{\hat\cY} \bbC(h^{-1}(\hat y))\big[\bH^\top\bE\bH+\bH^\top\nabla_{\hat y}\hat\bw^{\bE}(\hat y)\bH\big]d\hat y\\
&= \bH^{-1}\dashint_{\hat\cY} \bbC(\hat y)\big[\bE+\nabla_{\hat y}\hat\bw^{\bE}(\hat y)\big]d\hat y\, \bH^{-\top}\\
&=\bH^{-1}(\bbC^0\bE)\bH^{-\top}.
\end{align*}
\end{proof}

\begin{figure}
\includegraphics[height=2in]{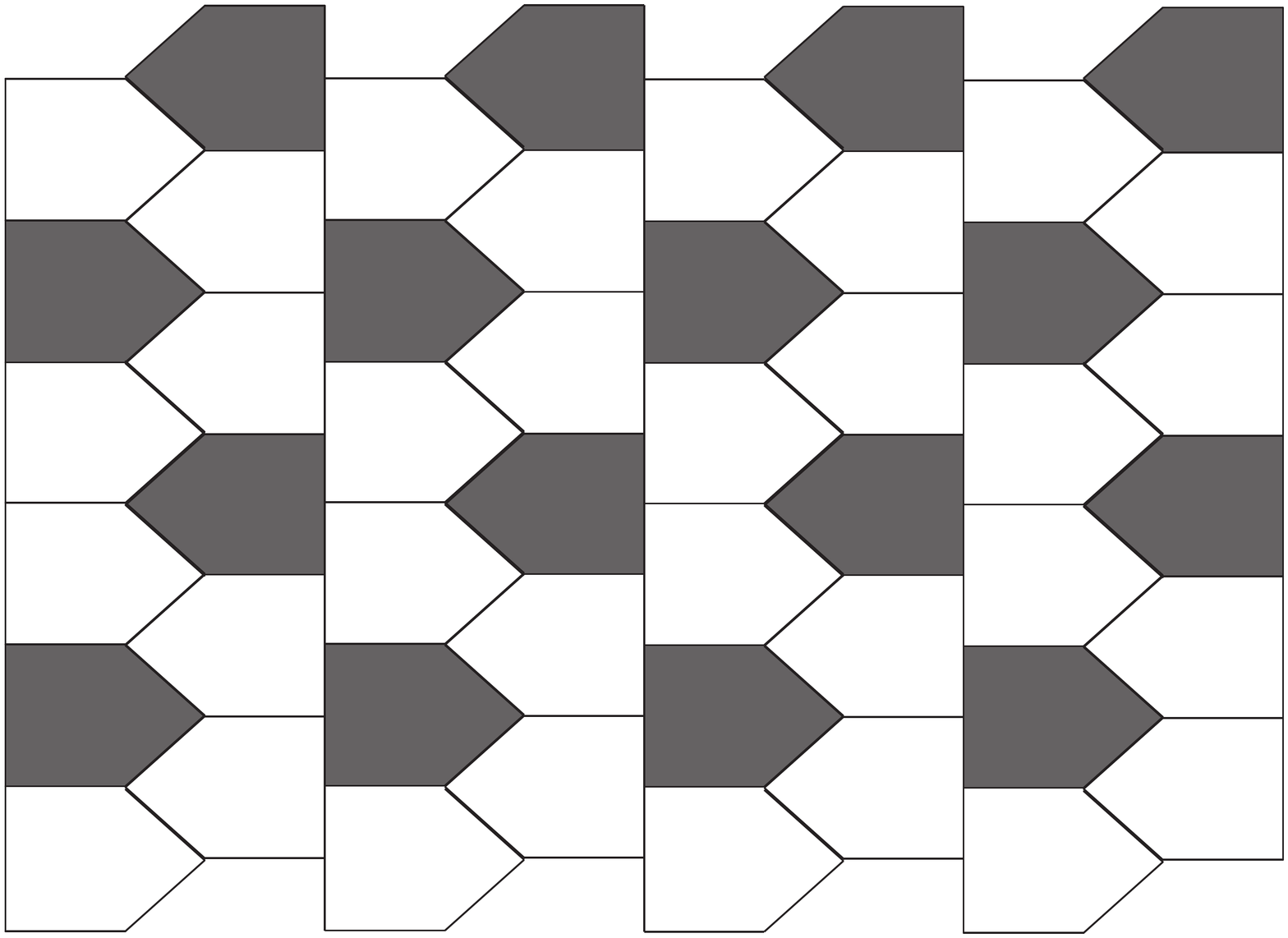}
\includegraphics[height=2in]{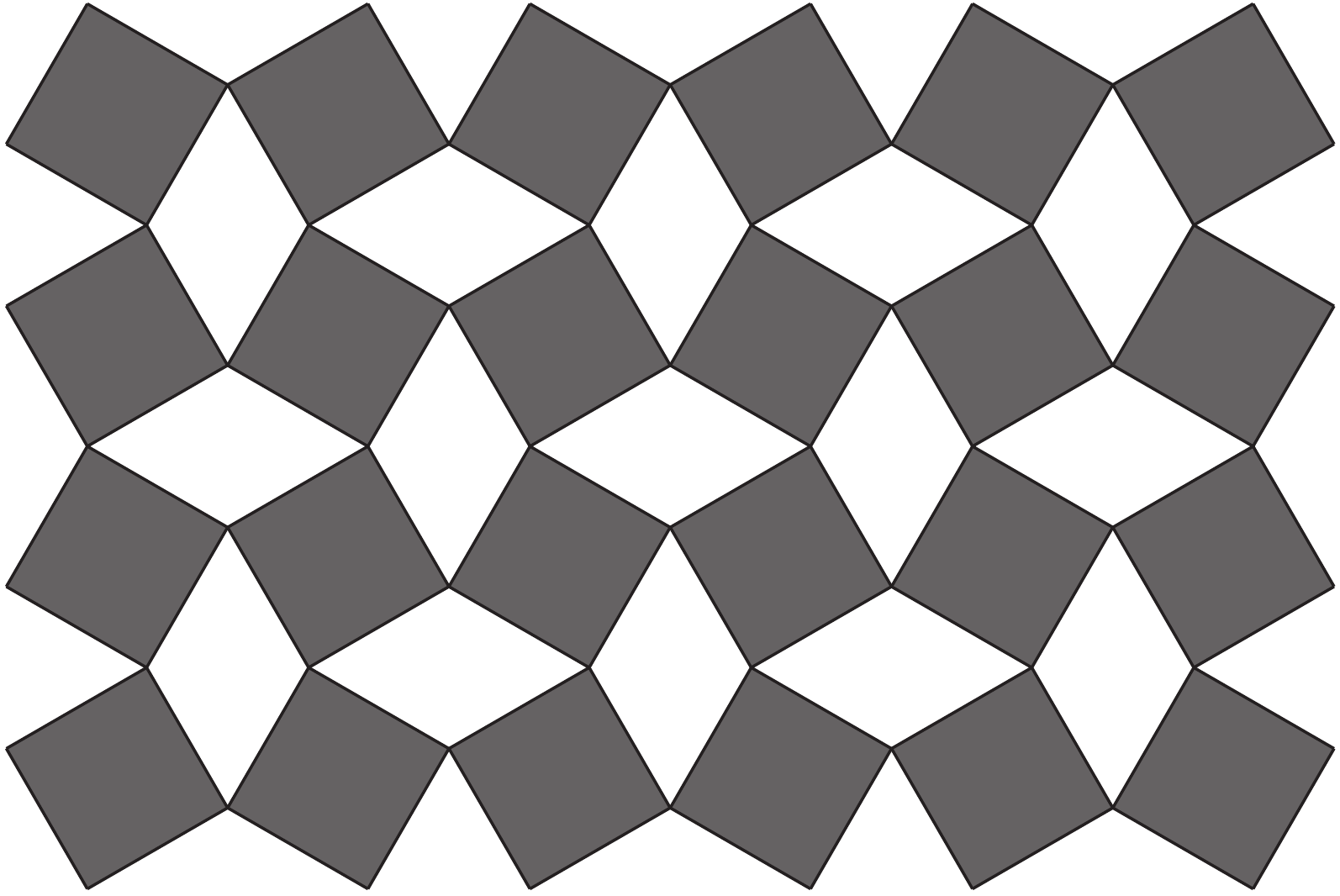}
\put(-324,-15){(a)}
\put(-114,-15){(b)}
\put(-146,34){$\bullet$}
\put(-140,30){$p$}
\put(-92,52){$\bullet$}
\put(-88,48){$q$}
\put(-313.5,50){$\ba$}
\thicklines
\put(-328.8,46.0){\rotatebox[origin=c]{90}{$\vector(1,0){33}$}}
\caption{Examples of periodic elastic structures containing symmetries that are not covered by previous results similar to Theorem~\ref{thmmain}.  In these examples, the white and grey regions should be viewed as different isotropic materials. (a) An example of a periodic elastic structure having a symmetry consisting of the translation $\ba$ together with a reflection in which neither the translation nor the reflection by themselves symmetries of the elastic structure. (b) This periodic elastic structure has amongst its symmetries reflections about the vertical and horizontal lines passing through $p$ and rotations about the point $q$ by $\pi/2$ radians.}
\label{figtranrot}
\end{figure}

Equation \eqref{mainresult} is the condition for $\bH$ to be in the material symmetry group of $\bbC^0$.  Thus, Theorem~\ref{thmmain} says that the gradient of every symmetry of the periodic elastic structure is a material symmetry of the macroscopic elasticity tensor.  For a detailed discussion of material symmetry see, for example, Gurtin, Fried, and Anand \cite{GFA} or Epstein and El$\dot{\text{z}}$anowski \cite{EE}.  

Notice that Theorem \ref{thmmain} does not say that every element of the material symmetry group of $\bbC^0$ is generated by a symmetry of the periodic elastic structure.  In fact, there are material symmetries that are not generated by symmetries of the periodic elastic structure.  This will be demonstrated with an example in the next section.

Although the above analysis was carried out in the context of linear elasticity, similar calculations hold for any first-order transport law.  In that case, a mobility tensor $\bM$ would be defined on a unit cell and \eqref{microsym} would be replaced by
$$
\bM(z)=\bH\bM(h^{-1}(z))\bH^\top\rfa z\in\cE.
$$
For example, the physical processes of diffusion, heat conduction, and dielectric induction have this kind of structure.

Results similar to Theorem~\ref{thmmain} were established by Jikov, Kozlov, and Oleinik \cite{JKO} and by Alexanderian, Rathinam, and Rostamian \cite{ARR}.  To accurately compare these results with Theorem~\ref{thmmain}, take the Euclidean space $\cE$ to be $\Real^n$, in which case points can be considered as vectors and vice versa.  Jikov, Kozlov, and Oleinik proved Theorem~\ref{thmmain} for symmetries of the form
\beqn\label{orthsym}
z\mapsto \bQ z\rfa z\in\Real^3,
\eeqn
where $\bQ$ is an orthogonal linear mapping and Alexanderian, Rathinam, and Rostamian also considered symmetries of this form in the context of the homogenization of diffusive random media.  Thus, these previous results only considered symmetries that are rotations or reflections about a single point---in particular, the origin.  Whereas Theorem~\ref{thmmain} also includes symmetries about different points, symmetries consisting of an orthogonal linear mappping together with a translation, and symmetries involving unimodular linear mappings.  

A periodic elastic structure with a symmetry consisting of a translation and reflection is depicted in Figure~\ref{figtranrot}(a).  In this figure consider the white and grey regions to consist of different isotropic elastic materials.  Notice that a reflection about any vertical line is not a symmetry of the elastic structure, however the transformation consisting of the translation $\ba$ together with the reflection about the vertical line aligned with $\ba$ is a symmetry of the periodic elastic structure.  Figure~\ref{figtranrot}(b) is an example of a periodic elastic structure that contains symmetries of the form 
$$
z\mapsto z_\circ+\bH (z-z_\circ)\rfa z\in\cE
$$
for different $\bH$ associated with different points $z_\circ$.  Namely, there are reflection symmetries about point $p$ and a rotation about the point $q$ by $\pi/2$ radians is a symmetry.  An example of a periodic elastic structure that has a symmetry where $\bH$ is not orthogonal is described in the next section.

When $\bbC$ takes only a finite number of values, the elastic structure is made up of a finite number of materially uniform constituents.  In this case, two factors determine the symmetries of such an elastic structure and, hence, the material symmetry of $\bbC^0$:  the material symmetries of the constituents and their arrangement within the periodic structure.  This is illustrated in the next section.  Parnell and Abrahams \cite{PA} realized this in the context of a three dimensional elastic structure made of two constituents.

\section{Examples}\label{sectEx}

In this section several examples of periodic elastic structures are given.  For each example, the symmetries of the periodic elastic structure are mentioned together with the resulting material symmetries of the macroscopic elasticity tensor guaranteed by Theorem~\ref{thmmain}.  

For simplicity, consider $\cE$ to be the three-dimensional space $\Real^3$ so that points and vectors are interchangeable and use the standard basis $\be_1,\be_2, \be_3$.  Given a unit vector $\ba$ and an angle $\phi$, let $\bR^\phi_\ba$ be the right-handed rotation through the angle $\phi$ about the axis in the direction of $\ba$.  Notice that $-\bR^\pi_\ba$ is the reflection with respect to the plane through the origin with normal $\ba$.

Since the macroscopic elasticity tensor $\bbC^0$ is given by an integral expression, if $\bbC$ is changed on a set of measure zero, then the macroscopic elasticity tensor remains unchanged.  For this reason, in the following examples the microscopic elasticity tensor $\bbC$ is only specified up to a set of measure zero.  Similarly, the unit cell $Y$ which defines the periodic elastic structure only tiles $\Real^3$ up to a set of measure zero.\\

\begin{figure}
\begin{center}
\includegraphics[width=1.92in]{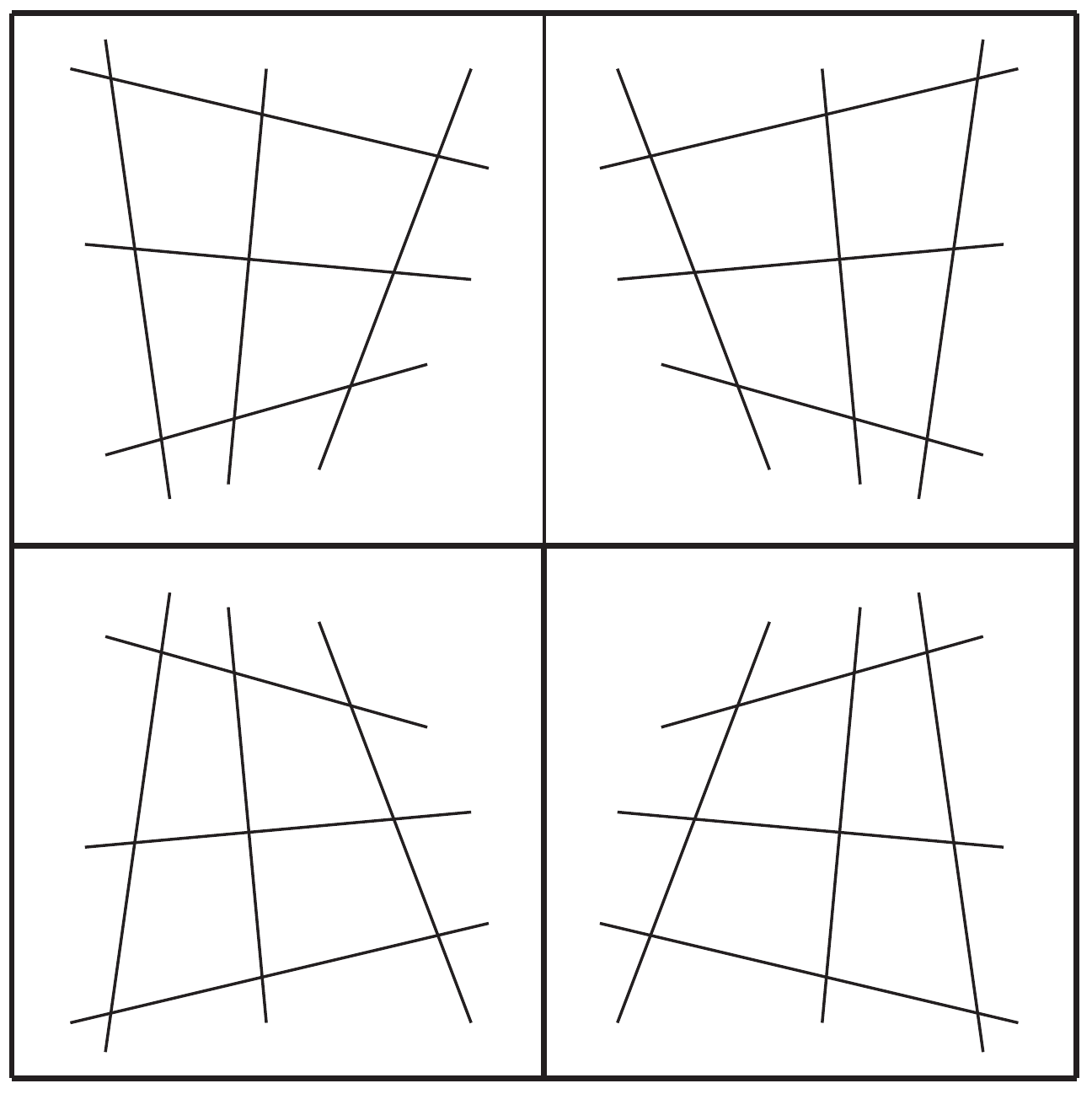}
\end{center}
\caption{A depiction of a cross-section of the unit cell described in the example involving orthotropic symmetry.  The cross-section is taken parallel to one of the coordinate axes.  The anisotropy of the elasticity tensor, depicted as crisscrossed lines, in the different quadrants are related through reflections.}
\label{aniso}
\end{figure}

\noindent \textbf{Orthotropic symmetry:}
It is possible for the elasticity tensor $\bbC$ associated with the periodic elastic structure to be fully anisotropic at each point and yet the macroscopic elasticity tensor can have orthotropic symmetry, meaning that
$$%\beqn\label{orthotropicsym}
\bbC^0(\bE)=\bH[\bbC^0(\bH^\top\bE\bH)\big]\bH^\top\rfa \bE\in\Sym\ \text{and}\ \bH = -\bR^\pi_{\be_1}, -\bR^\pi_{\be_2},-\bR^\pi_{\be_3}.
$$%\eeqn
To obtain this, the anisotropies of the elastic structure are arranged in a particular way.  Begin by considering an anisotropic elasticity tensor $\bbC^\text{an}$ and the unit cell $\cY=(-1,1)^3$.  To define $\bbC$ on $\cY$, start by setting $\bbC$ equal to $\bbC^\text{an}$ on the positive octant $(0,1)^3$.  Then, recalling \eqref{notation}, set $\bbC$ equal to $\bbC^\text{an}_{-\bR_{\be_1}^\pi}$ on $(-1,0)\times(0,1)\times(0,1)$.  Next, define $\bbC(y_1,y_2,y_3)$ with $y_1\in(-1,0)\cup(0,1)$, $y_2\in(-1,0)$, and $y_3\in(0,1)$ by $\bbC(y_1,y_2,y_3)=\bbC_{-\bR_{\be_2}^\pi}(y_1,-y_2,y_3)$.  Finally, define $\bbC(y_1,y_2,y_3)$ for negative $y_3$-coordinate by requiring that $\bbC(y_1,y_2,y_3)=\bbC_{-\bR_{\be_3}^\pi}(y_1,y_2,-y_3)$.  Since reflections about perpendicular planes commute, it is readily shown that the mappings
$$
z\mapsto-\bR^\pi_{\be_i} z \qquad \text{for all}\ z\in\Real^3,\, i=1,2,3
$$
are symmetries of the periodic elastic structure $(\cY,\bbC)$.  See Figure~\ref{aniso} for a depiction of the microstructure associated with $\bbC$.  It follows from Theorem~\ref{thmmain} that $-\bR^\pi_{\be_i}$, for $i=1,2,3$, are material symmetries of $\bbC^0$.\\

\begin{figure}
\includegraphics[width=1.92in]{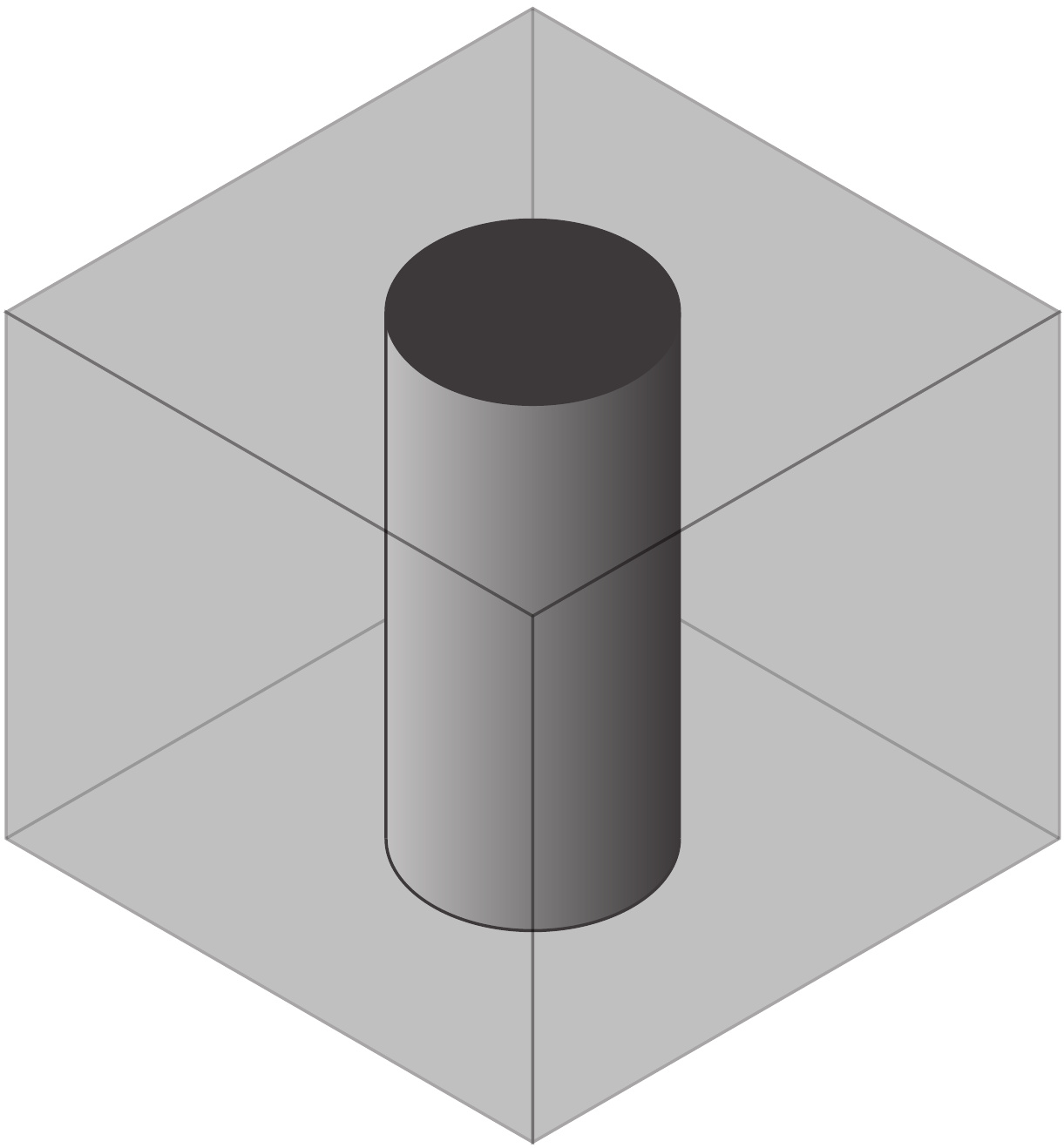}
\includegraphics[width=1.92in]{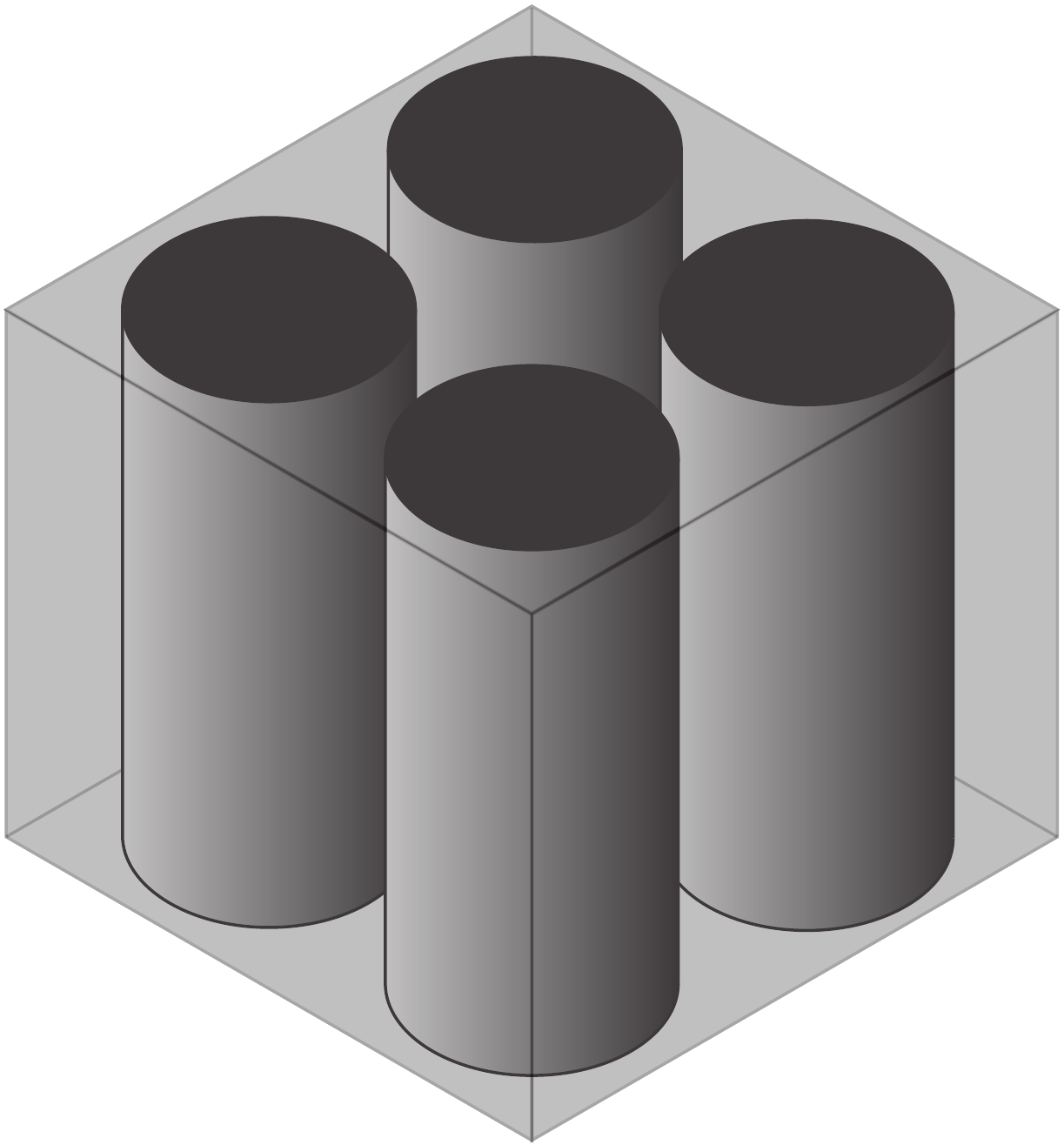}
\includegraphics[width=1.92in]{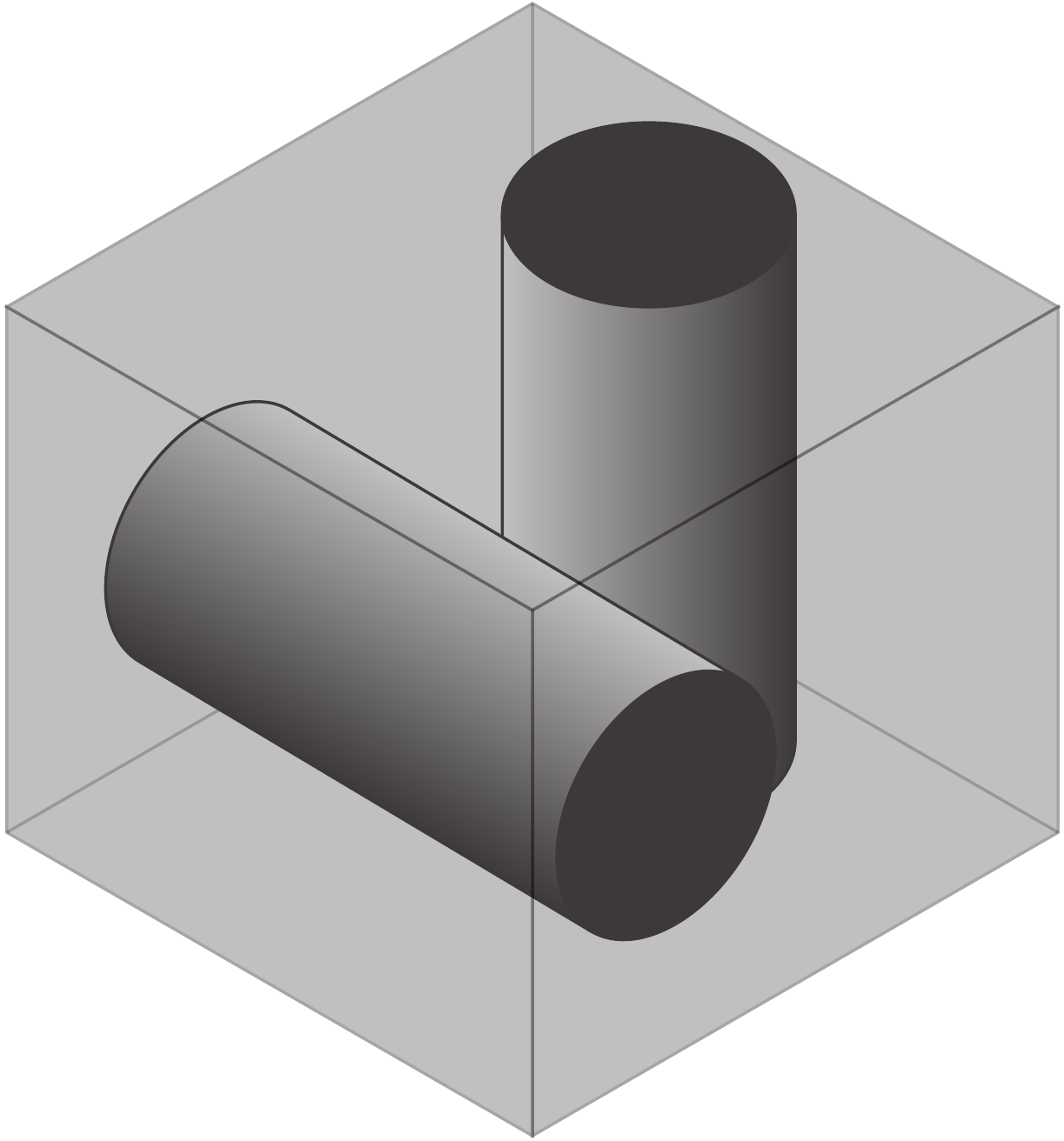}
\put(-361,-15){(a)}
\put(-218,-15){(b)}
\put(-75,-15){(c)}
\caption{The unit cell $Y=(-1,1)^3$ consisting of a matrix embedded with fibers in different configurations.  If the matrix and the fibers have tetragonal symmetry, then the associated macroscopic elasticity tensor has tetragonal symmetry. (a) The unit cell contains one fiber whose axis aligns with the $y_3$-axis.  (b) The unit cell contains four fibers oriented in the direction of the $y_3$-axis.  (c) The unit cell contains two fibers that are oriented in orthogonal directions.}
\label{figtetsym}
\end{figure}

\noindent \textbf{Tetragonal symmetry:}
Suppose the unit cell $\cY=(-1,1)^3$ consists of two different isotropic materials, one representing fibers and the other a matrix in which the fibers are embedded.  Suppose that there is one fiber that occupies a cylinder whose axis is the $y_3$-axis; see Figure~\ref{figtetsym}(a).  Notice that the mappings
$$
z\mapsto\bH z\qquad \text{for all}\ z\in\Real^3,\, \bH=\bR_{\be_1}^\pi,\ \bR_{\be_3}^{\pi/2}
$$
are symmetries of the periodic elastic structure.  To have $\bR_{\be_1}^\pi$ and $\bR_{\be_3}^{\pi/2}$ as material symmetries is sufficient for tetragonal symmetry.

The same symmetry persists if there are several fibers appearing in the unit cell arranged as depicted in Figure~\ref{figtetsym}(b).  In general the microstructures generated by the unit cells depicted in Figures~\ref{figtetsym}(a) and (b) are not the same.  However, if the radii of the fibers and the spacing of the fibers relative to the width of the unit cell are chosen appropriately, then they do generate the same microstructure.  As long as the material symmetry groups of the matrix and fibers contain $\bR_{\be_1}^\pi$ and $\bR_{\be_3}^{\pi/2}$, the result still holds.  Thus, if the fibers are transversely isotropic and the matrix is isotropic, then the macroscopic elasticity tensor has tetragonal symmetry.  This type of periodic elastic structure was used by Ptashnyk and Seguin \cite{PS} in modeling plant cell walls.  The macroscopic elasticity tensor also has tetragonal symmetry when the fibers are oriented in orthogonal directions, as depicted in Figure~\ref{figtetsym}(c).  To see this, let
\begin{align*}
& \{(y_1,y_2,y_3)\in \cY\ |\ (y_1-0.5)^2+y_3^2\leq 0.2^2\},\\
& \{(y_1,y_2,y_3)\in \cY\ |\ (y_1+0.5)^2+y_2^2\leq 0.2^2\}
\end{align*}
be the domains occupied by the two fibers in the unit cell.  It can be verified that the transformations
$$
z\mapsto \bH z\qquad \text{for all}\ z\in\Real^3,\, \bH=-\bR_{\be_2}^\pi,\ -\bR_{\be_3}^{\pi}
$$
and
$$
z\mapsto \be_1 + \bR_{\be_1}^{\pi/2}z\rfa z\in\Real^3
$$
are symmetries of the elastic structure and, thus, the macroscopic elasticity has tetragonal symmetry.\\

\noindent \textbf{Hexagonal symmetry:}
Let $R$ be a rhombus in $\Real^2$ whose short diagonal and sides have length $2$.  Consider the unit cell $\cY=(-1,1)\times R$.  As in the example involving tetragonal symmetry, assume that $\cY$ consists of a matrix and fibers and that the fibers occupy the domain
\beqn\label{cF}
\{(y_1,y_2,y_3)\in \cY\ |\ \text{dist}((y_2,y_3),R_i)<1,\ \text{where $R_i$ are the corners of $R$}\}
\eeqn
and the rest of $\cY$ is occupied by the matrix.  A cross-section of $Y$ at a constant $y_1$ value is shown in Figure~\ref{fighexsym}(a).   The elastic structure resulting from tiling $\Real^3$ with this unit cell has fibers arranged in a hexagonal bundle structure; see Figure~\ref{fighexsym}(b).  When the matrix and the fibers are isotropic, then the periodic elastic structure generated by this unit cell has the symmetry
$$
z\mapsto z_\circ + \bR_{\be_3}^{\pi/3}(z-z_\circ)\rfa z\in\Real^3,
$$
where $z_\circ$ is any point on the axis of one of the fibers. It follows that the corresponding macroscopic elasticity tensor has hexagonal symmetry.  

It is known that the form of an elasticity tensor possessing the material symmetry $\bR_{\be_3}^{\pi/3}$ is the same as that with transverse isotropy; see Gurtin \cite{Gurtin}.  Thus, the material symmetry group of the macroscopic elasticity tensor also contains $\bR_{\be_3}^\theta$ for all $\theta$.  However, there is no mapping $h$ of the form \eqref{isometry} whose gradient is $\bR_{\be_3}^\theta$ when $\theta$ is not a multiple of $\pi/3$.  Thus, this example shows that not every material symmetry of the macroscopic elasticity tensor is generated by a symmetry of the periodic elastic structure. \\

\begin{figure}
\hspace{.1in}\includegraphics[height=1.8in]{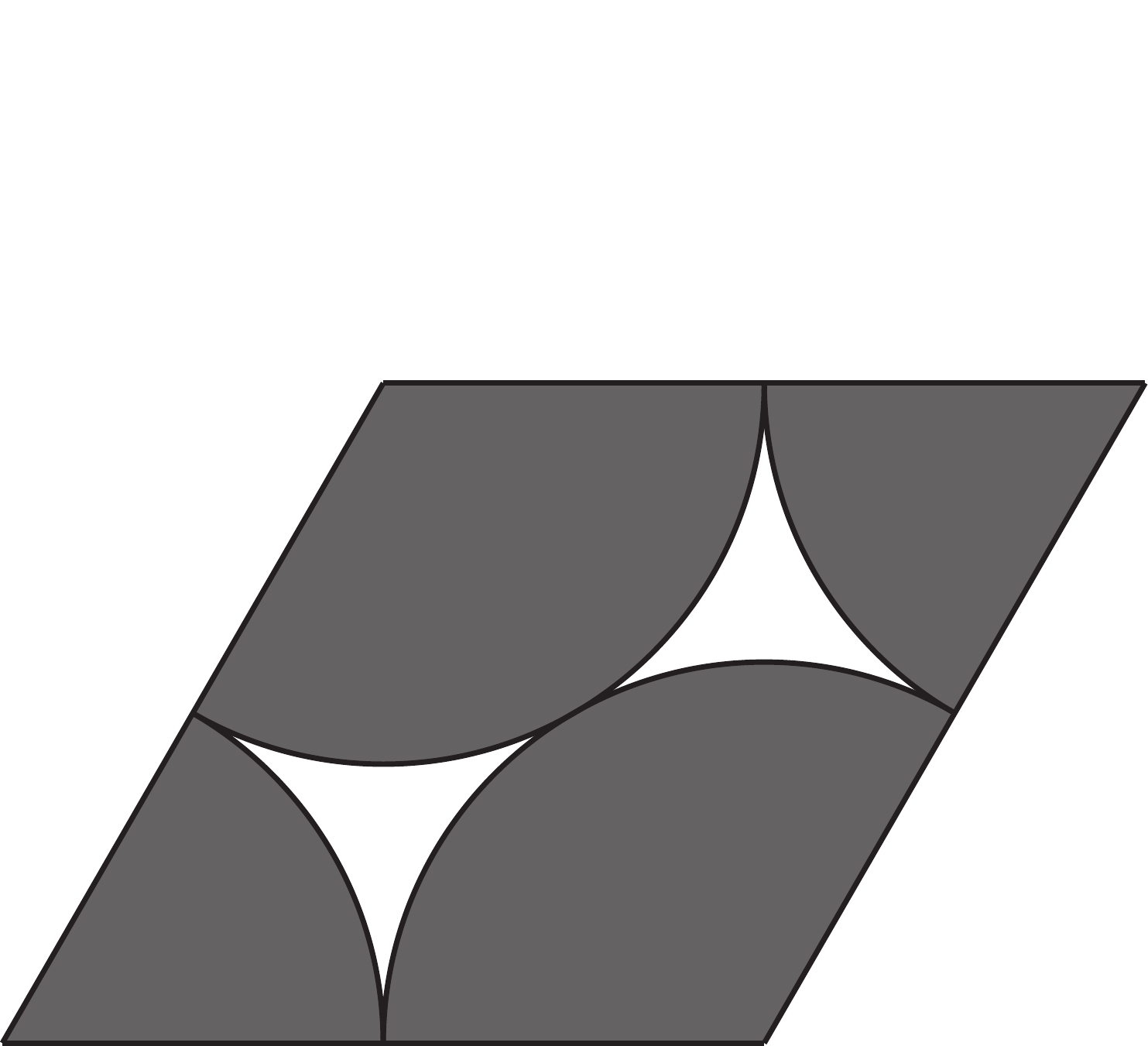}
\hspace{.1in}\includegraphics[height=1.3in]{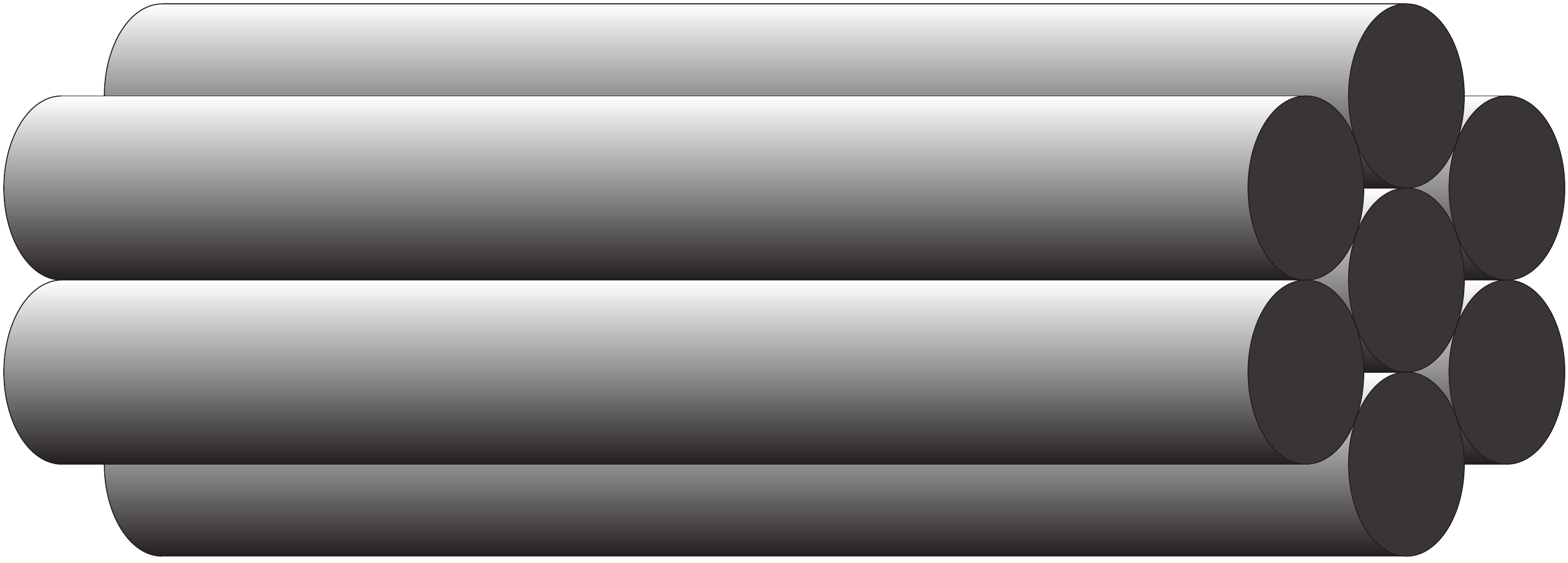}
\put(-359,-15){(a)}
\put(-150,-15){(b)}
\vspace{.05in}
\caption{(a) A cross-section with constant $y_1$ value of the unit cell $Y=(-1,1)\times R$, where $R$ is a rhombus in $\Real^2$ whose short diagonal and sides have length $2$.  The grey regions, defined in \eqref{cF}, are occupied by fibers and the white region is occupied by a matrix.  (b) A depiction of the periodic fiber structure obtained by tiling $\Real^3$ with the unit cell $Y=(-1,1)\times R$.  The fibers, which are shown, are arranged in a hexagonal bundle structure.  The surrounding matrix is not shown.  Here $\be_1$ is in the horizontal direction.}
\label{fighexsym}
\end{figure}

\noindent \textbf{In-plane fluidity:}
Suppose that the unit cell $Y=(-1,1)^3$ consists of two linearly elastic materials, one occupying the bottom half $(-1,1)^2\times(-1,0)$ of $Y$ and the other the top half $(-1,1)^2\times(0,1)$ of $Y$.  Moreover, assume that both materials have in-plane fluidity in the plane orthogonal to $\be_3$, meaning that their symmetry group is
$$
\cG = \{\bH\ \text{unimodular}\ | \ \bH\be_3=\be_3\ \text{and if}\ \be_3\cdot\ba=0, \text{then}\ \be_3\cdot\bH\ba=0\}.
$$
This symmetry group is useful in the modeling of lipid bilayers; see, for example, Deseri, Piccioni, and Zurlo \cite{DPZ} and Maleki, Seguin, and Fried \cite{MSF}.  It is readily verified that for all $\bH\in\cG$, the transformation
$$
z\mapsto \bH z\rfa z\in\Real^3
$$
is a symmetry of the periodic elastic structure and, hence, Theorem~\ref{thmmain} implies that the macroscopic elasticity tensor also has in-plane fluidity.

\section{Conclusion}\label{sectCon}

In this paper we showed a connection between the symmetries of a periodic linearly elastic structure and the material symmetry of the macroscopic elasticity tensor obtained by using homogenization theory.  Our result generalizes previous work by considering a larger class of symmetries on the microscopic scale.  Since the material symmetry group is a subgroup of the unimodular group, we conjecture that symmetries of the form \eqref{isometry} include all possible symmetries of a periodic elastic structure that lead to material symmetries of the macroscopic elasticity tensor.  However, as the example in Section~\ref{sectEx} involving hexagonal symmetry shows, not all material symmetries of the macroscopic elastic tensor are generated by symmetries of the periodic elastic structure.

The proof of Theorem~\ref{thmmain} relies on the formula \eqref{hombbC} for the macroscopic elasticity tensor.  Since homogenization of problems posed in domains with locally-periodic or random microstructures also results in explicit formulas for the macroscopic elasticity tensor, it is possible that with suitable modifications the arguments used here could yield a result similar to Theorem~\ref{thmmain} for these kinds of microstructures.  It would be interesting to investigate whether there is an analog of Theorem~\ref{thmmain} in the case of nonlinear elasticity.

\section{Appendix}

In this Appendix we present the details of the derivation of the macroscopic elasticity tensor $\bbC^0$ given in \eqref{hombbC}.

For the derivation, it is assumed that
\begin{enumerate}
\item there is a strictly positive $M$ such that $|\bbC(y)|\leq M$ for almost every $y\in Y$,
\item the function $y\mapsto \bbC(y)$ is measurable,
\item there is a strictly positive $\alpha$ such that $\alpha|\bE|^2\leq\bE\cdot\bbC\bE$ for all $\bE\in \Sym$,
\item $\bA\cdot\bbC\bB = \bA\cdot\bbC\big[\frac{1}{2}(\bB + \bB^\top)\big]=\big[\frac{1}{2}(\bA + \bA^\top)\big]\cdot\bbC\bB$ for all linear mappings $\bA$ and $\bB$, and
\item $\bD\cdot\bbC\bE=\bE\cdot\bbC\bD$ for all $\bD,\bE\in\Lin$.
\end{enumerate}

Let $\Omega$ be a reference configuration of the elastic material with elasticity tensor $\bbC^\ve$ defined in \eqref{bbCscale}, and assume that $\Omega$ is open and bounded with Lipschitz boundary so that locally the boundary can be represented as the graph of a Lipschitz function.  Let $\Gamma_1$ and $\Gamma_2$ be disjoint subsets of $\partial\Omega$ with $\partial\Omega=\Gamma_1\cup\Gamma_2$, and assume a zero displacement boundary condition on $\Gamma_1$ and the traction on $\Gamma_2$ is given by $\bt$.  The resulting mixed boundary-value problem in elastostatics with body force $\bb$ is
\begin{equation}\label{mixedpro}
\left \{
\begin{array}{ll}
\text{div}(\bbC^\ve\nabla\bu^\ve)+\bb=\textbf{0}\qquad &\text{in}\ \Omega, \\[10pt]
\bu^\ve = \textbf{0}\qquad &\text{on}\ \Gamma_1, \\[10pt]
(\bbC^\ve\nabla\bu^\ve)\bn = \bt\qquad &\text{on}\ \Gamma_2,
\end{array}
\right.
\end{equation}
where $\bn$ is the exterior unit-normal to $\Omega$.  Notice that by the assumed properties of $\bbC$, only the symmetric part of $\nabla\bu^\ve$ is relevant in \eqref{mixedpro}.  By the Lax--Milgram theorem \cite{OSY} a unique solution of \eqref{mixedpro} exists in the space
\beqn\label{defH}
\cH = \{\bu\in H^1(\Omega,\cV)\ |\ \bu = \textbf{0}\ \text{on}\ \Gamma_1\}
\eeqn
provided that $\bb\in L^2(\Omega,\cV)$ and $\bt\in L^2(\Gamma_2,\cV)$.  Moreover, the solutions $\bu^\ve$ of \eqref{mixedpro} are bounded in the $H^1$-norm independent of $\ve$ and, hence, there is a $\bu\in H^1(\Omega,\cV)$ such that up to a subsequence
\beqn
\label{strconv}\text{$\bu^\ve$ converges weakly to $\bu$ in $H^1(\Omega,\cV)$};
\eeqn
see \cite{OSY}.  The goal of homogenization is to find an equation that characterizes $\bu$.  

To derive the macroscopic equations associated with \eqref{mixedpro}, the notion of two-scale convergence first introduced by Nguetseng \cite{Ngu} and further developed by Allaire \cite{All} is applied using the following function spaces.

\begin{itemize}
\item $L^2_\text{per}(Y,\cV)$ is the set of all functions $\bpsi$ from $\cE$ to $\cV$ that are $Y$-periodic and $\int_Y|\bpsi(y)|^2\,dy$ is finite.
\item $H^1_\text{per}(Y,\cV)$ is the completion with respect to the $H^1$-norm of the space of smooth functions from $\cE$ to $\cV$ that are $Y$-periodic.
\item $\cW_\text{per}(Y,\cV)$ consists of equivalence classes of $H^1_\text{per}(Y,\cV)$ where two functions are equivalent if they differ by a constant vector.
\item $L^2(\Omega, \cW_\text{per}(Y,\cV))$ consists of all functions $\bpsi$ from $\Omega$ to $\cW_\text{per}(Y,\cV)$ such that $$\int_\Omega \|\bpsi(x,\cdot)\|_{H^1(Y,\cV)}^2dx$$ is finite.
\item $C(\overline\Omega,L^2_\text{per}(Y,\cV))$ consists of functions $\bpsi:\overline\Omega\times\cE\rightarrow\cV$ such that for all $x\in\overline\Omega$ the function $y\mapsto\bpsi(x,y)$ is in $L^2_\text{per}(Y,\cV)$ and the function $x\mapsto \bpsi(x,\cdot)$ from $\overline\Omega$ to $L^2_\text{per}(Y,\cV)$ is continuous.

\end{itemize}

\noindent Two-scale convergence, in the context presented here, is defined as follows.

\begin{definition}
A sequence of functions $\bv^\ve\in L^2(\Omega,\cV)$ two-scale converges to $\bv^0\in L^2(\Omega\times\cY,\cV)$ if for all $\bpsi\in C(\overline\Omega,L^2_\text{\rm per}(Y,\cV))$,
\beqn\label{twoscaledef}
\lim_{\ve\rightarrow 0}\int_\Omega \bv^\ve(x)\cdot\bpsi\big(x,\frac{x-q}{\ve}+q\big)dx = \int_\Omega\dashint_Y \bv^0(x,y)\cdot\bpsi(x,y)dydx,
\eeqn
where the symbol $\dashint$ denotes the average integral.
\end{definition}

\noindent By using the change of variables $x\mapsto x-q$, one can see that this definition is equivalent to the standard definition \cite{All,Ngu} involving a vector space rather than a Euclidean point space.  Associated with two-scale convergence is the following compactness result \cite{All,Ngu}.

\begin{theorem}\label{thmcomp}
If $\,\bv^\ve\in H^1(\Omega,\cV)$ is a sequence of functions that weakly converges to $\bv^0\in H^1(\Omega,\cV)$, then there is a $\bv^1\in L^2(\Omega,\cW_\text{\rm per}(Y,\cV))$ such that up to a subsequence
$$
\nabla\bv^\ve\ \text{converges two-scale to}\ \nabla\bv^0+\nabla_y\bv^1.
$$
\end{theorem}

%\noindent Following the tradition in homogenization theory, the symbol $\nabla_y$ denotes the gradient taken with respect to the microscopic variable $y$.  
\noindent It follows from Theorem~\ref{thmcomp} that \eqref{strconv} implies that there is a $\bw\in L^2(\Omega,\cW_\text{per}(Y,\cV))$ such that up to a subsequence
\begin{align}
\label{tsconv}\nabla\bu^\ve\ \text{converges two-scale to}\ \nabla\bu+\nabla_y\bw.
\end{align}

Define $\bv^\ve(x)=\bv^0(x)+\ve \phi(x)\bv^1(\frac{x-q}{\ve}+q)$ for $x\in\Omega$, where $\bv^0\in\cH\cap C^\infty(\Omega,\cV)$, $\phi\in C^\infty_0(\overline\Omega)$, and $\bv^1\in H^1_{\rm per}(Y,\cV)$.  Notice that 
\beqn
\label{vwconv}\text{$\bv^\ve$ converges weakly to $\bv^0$ in $H^1(\Omega,\cV)$}.
\eeqn
Multiply \eqref{mixedpro}$_1$ by $\bv^\ve$, integrate over $\Omega$, and then integrate by parts to obtain
\begin{multline}\label{hom0}
-\int_\Omega \nabla\bu^\ve(x)\cdot\bbC^\ve(x)\big[\nabla\bv^0(x)+\ve\bv^1\big(\frac{x-q}{\ve}+q\big)\otimes\nabla\phi(x)+\phi(x)\nabla_y\bv^1\big(\frac{x-q}{\ve}+q\big)\big]dx\\
 + \int_{\Gamma_2}\bt(x)\cdot\bv^\ve(x)\,dx+\int_\Omega\bb(x)\cdot\bv^\ve(x)\,dx=0.
\end{multline}
Notice that the mapping
$$
(x,y)\mapsto \bbC(y)\big[\nabla\bv^0(x)+\phi(x)\nabla_y\bv^1(y)\big]
$$
is a suitable test function in the definition of two-scale convergence and that 
\beqn\label{neededconv}
\ve\int_\Omega \nabla\bu^\ve(x)\cdot\bbC^\ve(x)\big[\bv^1\big(\frac{x-q}{\ve}+q\big)\otimes\nabla\phi(x)\big]dx\longrightarrow 0\qquad \text{as}\ \ve\rightarrow 0.
\eeqn
Taking the limit as $\ve\rightarrow 0$ in \eqref{hom0} and using \eqref{tsconv}, \eqref{vwconv}, and \eqref{neededconv} together with the symmetry of $\bbC$ yields
\begin{multline}\label{hom1}
-\int_\Omega\dashint_Y \bbC(y)\big[\nabla\bu(x)+\nabla_y\bw(x,y)\big]\cdot\big[\nabla\bv^0(x)+\phi(x)\nabla_y\bv^1(y)\big]dydx\\
 + \int_{\Gamma_2}\bt(x)\cdot\bv^0(x)\,dx+\int_\Omega\bb(x)\cdot\bv^0(x)\,dx=0.
\end{multline}
Taking $\bv^0=\textbf{0}$ in \eqref{hom1} and using the arbitrariness of $\phi\in C^\infty_0(\overline\Omega)$ yields
\begin{multline}\label{hom2}
\int_\cY \bbC(y)\big[\nabla\bu(x)+\nabla_y\bw(x,y)\big]\cdot\nabla_y\bv^1(y)\, dy = 0\\
\text{for all}\ \bv^1\in H^1_\text{per}(Y,\cV)\ \text{and almost every}\ x\in\Omega.
\end{multline}
It can be shown that given $\nabla\bu$, there is exactly one $\bw\in L^2(\Omega,\cW_\text{per}(\cY,\cV))$ that satisfies $\eqref{hom2}$; see, for example, \cite{OSY}.  Also, $\bw$ only depends on $x$ through the symmetric part of $\nabla\bu(x)$ and, moreover, $\bw$ depends linearly on the symmetric part of $\nabla\bu(x)$.  Thus, \eqref{hom2} can be reformulated as follows: for any $\bE\in \Lin$, there is a unique $\bw^\bE\in \cW_\text{per}(Y,\cV)$ such that
\begin{equation}\label{unitcellA}
\displaystyle\int_\cY \bbC(y)\big[\bE+\nabla_y\bw^\bE(y)\big]\cdot\nabla_y\bv(y)\, dy = 0 \rfa \bv\in H^1_\text{per}(\cY,\cV),
\end{equation}
which is called the unit cell problem.  

Next, take $\phi=0$ in \eqref{hom1} to obtain
\begin{multline}\label{hom2.5}
-\int_\Omega\dashint_Y \bbC(y)\big[\nabla\bu(x)+\nabla_y\bw(x,y)\big]\cdot \nabla\bv^0(x)\, dydx\\
 + \int_{\Gamma_2}\bt(x)\cdot\bv^0(x)\,dx+\int_\Omega\bb(x)\cdot\bv^0(x)\,dx=0.
\end{multline}
If we define
\beqn\label{hombbCA}
\bbC^0\bE := \dashint_\cY \bbC(y)\big[\bE+\nabla_y\bw^\bE(y)\big]\, dy\rfa \bE\in \text{Sym},
\eeqn
which when written in components relative to an orthonormal basis has the form
\beqn\label{hombbCcomp}
\bbC^0_{ijkl} = \dashint_\cY(\bbC_{ijkl}(y)+\sum_{p,q=1}^n\bbC_{ijpq}(y)\partial_{y_q}\bw^{\be_k\otimes\be_l}_{p}(y))\, dy,
\eeqn
then \eqref{hom2.5} can be rewritten as
\beqn\label{hom3}
-\int_\Omega \bbC^0\nabla\bu(x)\cdot\nabla\bv^0(x)\,dx + \int_{\Gamma_2}\bt(x)\cdot\bv^0(x)\,dx+\int_\Omega\bb(x)\cdot\bv^0(x)\,dx=0.
\eeqn
Since \eqref{hom3} holds for all $\bv^0\in\cH\cap C^\infty(\Omega,\cV)$, this equation is the weak formulation of
\begin{equation}\label{hommixedpro}
\left \{
\begin{array}{ll}
\text{div}(\bbC^0\nabla\bu)+\bb=\textbf{0}\qquad &\text{in}\ \Omega, \\[10pt]
\bu = \textbf{0}\qquad &\text{on}\ \Gamma_1, \\[10pt]
(\bbC^0\nabla\bu)\bn = \bt\qquad &\text{on}\ \Gamma_2,
\end{array}
\right.
\end{equation}
which are the macroscopic equations associated with the microscopic system \eqref{mixedpro}.  This justified calling $\bbC^0$ the macroscopic elasticity tensor.  A unique solution of \eqref{hommixedpro} exists in $\cH$ by the Lax--Milgram theorem \cite{OSY}.  Numerically, it is more efficient to solve the system \eqref{hommixedpro} than \eqref{mixedpro} since the elasticity tensor in \eqref{hommixedpro} is constant, while in \eqref{mixedpro} the elasticity tensor can have rapid oscillations.  

An analogous homogenization argument can be carried out for elasticity problems involving different boundary conditions or including an inertial term.  However, the resulting formula \eqref{hombbC} for the macroscopic elasticity tensor, which is of primary interest here, would remain unchanged.\\

\noindent \textbf{Acknowledgement:} M.~Ptashnyk and B.~Seguin gratefully acknowledge the support of the EPSRC First Grant EP/K036521/1 ``Multiscale modelling and analysis of mechanical properties of plant cells and tissues''.

%$$
%(\bbC_2)_{ijkl} = (\bQ\otimes\bQ)_{ijkl} = \bQ_{ij}\bQ_{kl}
%$$

%$$\mathbb{S}_{ijkl} = \frac{1}{2}(\delta_{ik}\delta_{jl}+\delta_{jk}\delta_{il})$$
%
%$$(\bA\boxtimes\bB)_{ijkl} = \bA_{ik}\bB_{jl}$$
%
%$$(\bbC_4)_{ijkl} = \bP_{ik}\bQ_{jl} + \bP_{il}\bQ_{jk} + \bQ_{ik}\bP_{jl}+\bQ_{il}\bP_{jk}$$
%
%$$(\bbC_5)_{ijkl} = \bQ_{ik}\bQ_{jl}+\bQ_{il}\bQ_{jk} - \bQ_{ij}\bQ_{kl}$$

%http://en.wikipedia.org/wiki/Tiling_by_regular_polygons#Archimedean.2C_uniform_or_semiregular_tilings}

\bibliographystyle{acm}
\bibliography{homogensym} 

\end{document}